%% file: main_long.tex
\documentclass[conference,letterpaper]{IEEEtran}
\include{main_defs}
\def \LONG{}

\begin{document}
\input{main.tex}

\clearpage
\appendices
\crefalias{section}{appsec}

\section{Postponed proofs of \texorpdfstring{\Cref{sec:aon}}{}}
\label{app:proofs_aon}
\input{appendix/postponed_proofs_aon}

\section{Postponed proofs of \texorpdfstring{\Cref{sec:mle}}{}}
\label{app:proofs_mle}
\input{appendix/postponed_proofs_mle}

\section{Useful lemmas}
\label{app:lemmas}
\input{appendix/appendix_useful_lemmas}

\newpage
\bibliographystyle{IEEEtran}
\bibliography{bib}

\end{document}

%% file: main_defs.tex
\usepackage[dvips]{graphicx}
\usepackage[cmex10]{amsmath}
\usepackage{amsthm}
\usepackage{bbm}
\usepackage[colorlinks=true,linkcolor=red,citecolor=blue]{hyperref}
\usepackage{cleveref}
\usepackage{booktabs}
\usepackage{url}
\usepackage{color}
\usepackage{bm}
\usepackage{amssymb}
\usepackage{mathtools}
\usepackage[utf8]{inputenc} 
\usepackage[T1]{fontenc}
\usepackage{url}
\usepackage{ifthen}
\usepackage{cite}

\DeclareMathOperator{\MMSE}{MMSE}
\DeclareMathOperator{\CMMSE}{C-MMSE}
\DeclareMathOperator{\CMSE}{C-MSE}
\DeclareMathOperator{\MSE}{MSE}
\DeclareMathOperator{\MLE}{MLE}
\DeclareMathOperator{\EP}{EP}
\DeclareMathOperator{\MEP}{MEP}
\DeclareMathOperator{\CMEP}{C-MEP}

\DeclareMathOperator{\LL}{L}
\DeclareMathOperator{\Cl}{C-\textit{l}}
\DeclareMathOperator{\CL}{C-L}

\DeclareMathOperator{\zerone}{0-1}
\DeclareMathOperator{\Top}{Top}
\DeclareMathOperator*{\argmin}{argmin}
\DeclareMathOperator*{\argmax}{argmax}

\DeclarePairedDelimiter\floor{\lfloor}{\rfloor}

\def\bx{\mathbf{x}}

\def\bU{\mathbf{U}}
\def\bV{\mathbf{V}}
\def\bX{\mathbf{X}}
\def\bY{\mathbf{Y}}
\def\bZ{\mathbf{Z}}

\def\reals{\mathbb{R}}
\def\P{\mathbb{P}}
\def\Q{{\mathbb Q}}
\def\E{{\mathbb E}}

\def\normal{\mathcal{N}}

\newtheorem{theorem}{Theorem}
\newtheorem{definition}{Definition}
\newtheorem{lemma}{Lemma}

%% file: main.tex
\title{On maximum-likelihood estimation in the all-or-nothing regime} 

\author{\IEEEauthorblockN{Luca Corinzia, and 
Paolo Penna}
\IEEEauthorblockA{
Department of Computer Science\\
ETH Zürich, Switzerland\\
Email: \{luca.corinzia,paolo.penna\}@inf.ethz.ch}
\and
\IEEEauthorblockN{Wojciech Szpankowski}
\IEEEauthorblockA{Department of Computer Science \\
Purdue University, USA\\
Email: szpan@purdue.edu}
\and 
\IEEEauthorblockN{Joachim M. Buhmann}
\IEEEauthorblockA{
Department of Computer Science\\
ETH Zürich, Switzerland\\
Email: jbuhmann@inf.ethz.ch}
}

\maketitle

\begin{abstract}
\input{sections/abstract}
\end{abstract}
\ifdefined\LONG \else
\textit{A full version of this paper is accessible at:}
\url{https://arxiv.org/pdf/21xx.xxxx.pdf} 
\fi
\section{Introduction}
\input{sections/introduction}

\section{Setting}
\input{sections/setting}

\section{Generalized all-or-nothing phenomenon}
\label{sec:aon}
\input{sections/generalized_aon}

\section{Maximum-likelihood estimation}
\label{sec:mle}
\input{sections/application_tensor_PCA}

\section{Conclusion}
\input{sections/conclusion}

\section*{Acknowledgement}
This work was supported in part by
NSF Center on Science of Information
Grants CCF-0939370 and NSF Grants CCF-1524312, CCF-2006440, CCF-2007238.

%% file: sections/abstract.tex
We study the problem of estimating a rank-1 additive deformation of a Gaussian tensor according to the \emph{maximum-likelihood estimator} (MLE). 
The analysis is carried out in the sparse setting, where the underlying signal has a support that scales sublinearly with the total number of dimensions.
We show that for Bernoulli distributed signals, the MLE undergoes an \emph{all-or-nothing} (AoN) phase transition, already established for the minimum mean-square-error estimator (MMSE) in the same problem.
The result follows from two main technical points: (i) the connection established between the MLE and the MMSE, using the first and second-moment methods in the constrained signal space, (ii) a recovery regime for the MMSE stricter than the simple error vanishing characterization given in the standard AoN, that is here proved as a general result.

%% file: sections/introduction.tex
A fundamental question in information theory, statistics, and machine learning is to establish the \emph{computational} limits of estimation problems and determine the \emph{statistical} limit as the inviolable benchmark for the same problem. The common picture that arises in many problems is given by the presence of \emph{phase transitions} where the behaviour of the optimal estimators changes abruptly with the variation of the parameter of the problem. Typically, at least two phase transitions are present: the statistical phase transition that establishes the limit of \emph{any} estimator, and the computational phase transition at higher signal strength that establishes the limit of \emph{tractable} estimators, despite the two can coincide.
In the sparse setting of many estimation problems, a rather different picture emerges, as the statistical and computational phase transitions are situated at different scales of the parameters space and such gap diverges in the limit of vanishing sparsity. Moreover, the statistical phase transition is characterized by the so-called \emph{all-or-nothing} phenomenon (AoN): below a critical signal strength, the recovery of the planted signal is impossible, above the threshold is possible and with vanishing error. Although the AoN is conjectured to extend to the behaviour of any optimal estimator, the analysis has been so far focused on the minimum mean-square-error (MMSE) estimation, that is typically a \emph{bulk} estimator and hence can be too coarse for specific applications in the sparse setting. Hence, it would be desirable to extend the analysis of this phenomenon to other estimators, like the maximum-likelihood estimator (MLE), which recently received attention for showing optimal performance in retrieving the planted signal \cite{jagannath2020statistical} in a tensor-PCA model.

\subsection{Contribution}
In this work, we provide new results on the AoN phenomenon in the sparse estimation setting, with the following main contributions:
\begin{itemize}
    \item In \Cref{th:rate}, we generalize the AoN phenomenon proved in \cite{niles2020all} for the additive Gaussian noise model and the MMSE to arbitrarily \emph{asymptotics} in the \emph{recovery} regime. The proof follows a conditional second-moment method argument \cite{szpankowski2011average} and extends the proof given in \cite{niles2020all} with a careful control of the asymptotics of the bounds. This result is of independent interest as more stringent conditions than the simple error vanishing characterization are needed in specific applications.
    \item As an application of the first result, we study the \emph{maximum-likelihood estimator} in the sparse tensor-PCA problem and show in \Cref{th:aon_MEP} that also this estimator undergoes a weak AoN transition.
    \item The proof of the latter results is of independent interest, as it exploits the relations between different estimators with first and second-moment methods, crucially introducing the analysis of estimators \emph{constrained} in the signal hypothesis space. As a side result, the weak AoN phenomenon is proved for the constrained MMSE in \Cref{th:aon_CMMSE}. 
\end{itemize}

\subsection{Related work}
The problem of high dimensional statistical estimation that we study here has received much attention recently, with considerable progress obtained in the last years in understanding planted matrix and tensor models. Early works on statistical and computational limits of estimation focused on dense problems where the signal effective dimensionality scales linearly with the problem's dimensionality. Examples include: (i) compressed sensing \cite{donoho2009message} and matrix-PCA \cite{dia2016mutual, deshpande2014information} where the approximated message passing (AMP) algorithms are introduced and demonstrated to match the statistical phase transition; (ii) the tensor-PCA extension \cite{richard2014statistical} where the statistical and computational transitions are currently separated by a gap, considering a wide range of algorithms, i.e. spectral \cite{montanari2016limitation, perry2020}, AMP \cite{lesieur2017statistical}, Sum-of-Square \cite{hopkins2015tensor} and gradient descent \cite{arous2020algorithmic, biroli2020iron}. Many of these works focused on the mean-square-error and (high dimensional, i.e., matrix or tensorial) posterior average estimator. Nonetheless, recently more attention has been given to other estimators, e.g., in \cite{jagannath2020statistical} where the vectorial maximum-likelihood estimator has been shown to reach optimal correlation with the planted signal in the tensor-PCA model. See \cite{zdeborova2016statistical} for a thorough review in the field.

In the sparse regime in which the hidden signal's dimensionality is sublinear to the problem's dimensionality, the AoN phenomenon emerges. This phenomenon have been shown recently to hold in a wide range of problems, i.e., for sparse linear regression \cite{reeves2019all}, sparse matrix-PCA \cite{barbier2020all} and sparse tensor-PCA \cite{niles2020all} according to the mean-square-error loss. However, to the best of our knowledge, only a few other works studied how the same AoN phenomenon extends to other estimators. Examples include \cite{corinzia2019exact, corinzia2020statistical}, where non-matching upper and lower bounds are provided for the transition of the vectorial-MLE in the sparse planted hypergraph problem (equivalent to sparse tensor-PCA up to a reparameterization of the dimensionality of the problem), and \cite{gamarnik2017high} where the AoN is proved in the sparse linear regression model for the vectorial-MLE estimator.

%% file: sections/setting.tex
We study the estimation problem with observations given by the Gaussian additive model
\begin{equation}
\label{eq:model}
\bY = \sqrt{\lambda} \bX + \bZ
\end{equation}
where the signal to be estimated $\bX \in \reals^n$ is corrupted by Gaussian noise $\bZ$, with the collection $\{Z_i\}_{i=1}^n \overset{iid}{\sim} \P_z = \mathcal{N}(0, 1)$. 
We further assume that the prior distribution of $\bX$, denoted by $\P_n$, is uniform and discrete with support $supp(\P_n) \subset \mathcal{S}_{n-1}$, where $\mathcal{S}_{n-1}$ is the unit sphere in $\reals^n$. 
We denote by $M_n = |supp(\P_n)|$ the cardinality of the support of $\bX$ and by $\Q_{y|x} (\bY|\bX) = \mathcal{N}(\bY | \sqrt{\lambda} \bX , \bm{1}_{n\times n})$ the conditional distribution of $\bY$ given $\bX$, where $\bm{1}_{n\times n}$ is the identity matrix. 
We hence define $$\Q_{\lambda, n}(\bY) = \E_{\P_n}[\Q_{y|x}(\bY|\bX)]$$ as the distribution over the observations $\bY$, highlighting the respective signal-to-noise-ratio (snr) $\lambda$ and the problem dimension $n$ for convenience.
\begin{definition}
\label{def:estimator}
Let us denote as $\bar{\mathcal{S}}_{n-1} = \{ \bX \in \reals^n \colon \|\bX\| \le 1\}$ the unit ball. For any set $A$ we use the short notation $\min_{\hat{\bX}(\bY) \in A}$ for $\min_{\hat{\bX}: \bY \to \hat{\bX}(\bY) \in A}$ and analogously for other operators. For a generic bounded loss function $\LL$ we define the respective optimal estimator as
\begin{equation*}
\bX_{\LL}(\bY) = \argmin\limits_{ \hat{\bX}(\bY) \in \bar{\mathcal{S}}_{n-1}} \E\left[ \LL(\bX, \hat{\bX}(\bY)) \right].
\end{equation*}
The minimum loss achieved by such estimator is the quantity
\begin{equation*}
\LL_n(\lambda) = \E\left[ \LL(\bX, \bX_{\LL}(\bY)) \right].
\end{equation*}
We further define the best estimator that solves the optimization problem constrained to set $A$ with $supp(\P_n) \subset A \subset \bar{\mathcal{S}}_{n-1}$ as
\begin{equation*}
\bX_{\CL}(\bY) = \argmin\limits_{\hat{\bX}(\bY) \in A} \E\left[ \LL(\bX, \hat{\bX}(\bY)) \right],
\end{equation*}
and the respective minimum loss achieved as 
\begin{equation*}
\CL_n(\lambda) = \E\left[ \LL(\bX, \bX_{\CL}(\bY)) \right].
\end{equation*}
\end{definition}
We can easily observe that since $A \subset \bar{\mathcal{S}}_{n-1}$ then $\CL_n(\lambda) \ge \LL_n(\lambda)$ for any loss function $\LL$. 

\begin{definition}
\label{def:aon}
Denote by $\LL$ a bounded loss function $\LL \colon \bar{\mathcal{S}}_{n-1}^{\otimes 2} \to \reals_+$ with $\LL(\bX, \bX) = 0$. 
Denote by $c > 0$ the optimal error obtained by the estimator \emph{independent} on the observations $\bY$ as $c = \lim_{n \to \infty} \LL_n(0)$.
The estimation problem with observations given by \Cref{eq:model} with prior $\P_{n}$ satisfies the \emph{all-or-nothing} phenomenon (AoN) with respect to the loss $\LL$, with recovery asymptotics at least $\tau_n \in o(1)$ and critical snr $\lambda_n$ if 
\begin{equation}
\label{eq:aon}
\lim_{n \to \infty} \LL_n(\beta \lambda_n) =
\begin{cases}
\begin{aligned}
&\ c \ \ &\beta < 1 \\
&\ 0 \ \ &\beta > 1
\end{aligned}
\end{cases}
\end{equation}
and moreover $$\LL_n(\beta \lambda_n) \in o(\tau_n)$$ for $\beta > 1$, where $\beta$ is a constant independent on $n$.
\end{definition}
Intuitively, in the AoN phenomenon the estimation is impossible for a normalized snr smaller then a critical value $\beta_c = 1$, as the loss converges to the error achieved by an uninformative estimator, equivalent to the loss given by $\beta = 0$, while for higher snr the estimation is almost perfect, with error smaller then a given $\tau_n \to 0$. Note the difference of this definition from the one given in \cite{niles2020all} in the recovery regime $\beta > 1$. In  latter case the simpler condition $\LL_n(\beta \lambda_n) \to 0$ is given such that the asymptotics of the loss in the recovery regime in no further characterize. In the following we denote by $\tau_n$ a vanishing sequence such that $\tau_n \in o(1)$.

%% file: sections/generalized_aon.tex
In this section, we consider the mean-square-error (MSE), $\E[\|\bX - \hat{\bX}\|^2]$ where the expectation is taken with respect to $\P_n$ and $\P_z$, and $\hat{\bX} \coloneqq \hat{\bX}(\bY)$ is an estimator of the signal given the observation $\bY$. 
The MSE is minimized by the posterior average $\bX_{l_2} = \E[\bX|\bY]$, with the average taken respect to the posterior $\P_n(\bX | \bY)$ (see  \cite{cover1999elements} and \ifdefined\LONG \Cref{le:mmse}\else the appendix in the long version of the paper\fi). 
The minimum mean-square-error (MMSE) is then $\MMSE_n(\lambda) \coloneqq \E[\|\bX - \E[\bX | \bY]\|^2]$. This error is the minimum square-error achievable by any estimator that has access to the observations $\bY$. In the following, we denote by $D(p || q)$ the KL divergence between the distribution $p$ and $q$, and by $\pm o(\tau_n)$ a sequence $f_n \in o(\tau_n)$ that is respectively non-negative and non-positive. We further use the scaling $\lambda_n = 2 \log M_n$. A sufficient condition for having the AoN phenomenon is given by the property of the overlap rate function defined here.
\begin{definition}
For any $t \in [-1, 1]$ define the overlap rate function between two independent instances of the signal $\bX$ and $\bX'$ as
\begin{equation*}
r_n(t) = - \frac{1}{\log M_n} \log \P_n^{\otimes 2}[\langle\bX,\bX'\rangle \ge t].
\end{equation*}
where $\langle \bX, \bX' \rangle$ is the scalar product of two vectors.
\end{definition}
Intuitively, the rate function describes the rate of the exponential decay of the overlap $\langle \bX, \bX' \rangle$. 
The following theorem shows that a simple lower bound on the overlap rate function is sufficient to establish the AoN with recovery asymptotics $\tau_n$ if the latter is not too small.

\begin{theorem}
\label{th:rate}
For any $\epsilon > 0$ constant, if $\lambda_n^{-1/2 + \epsilon} \in o(\tau_n)$ and the overlap rate function $r_n(t)$ satisfies $$r_n(t) \ge \frac{2 t}{1 + t} - o(\tau_n),$$ then the probability $\P_n$ of the problem defined in \Cref{eq:model} satisfies the AoN in \Cref{def:aon} with recovery asymptotics at least $\tau_n$ according the mean-square-error.
\end{theorem}

\begin{proof}
The first part of \Cref{def:aon} related to \Cref{eq:aon} follows easily noting that the assumption given here is stricter than the one given in \cite{niles2020all}. We hence have to prove only that the stronger asymptotics holds in the recovery regime. The proof follows the steps of the proof in \cite{niles2020all} and mainly uses the widely known I-MMSE relation that relates the MMSE to the mutual information $I(\bX, \bY)$ and hence to the $D(\Q_{\lambda,n} || \Q_{0,n})$. It then uses the conditional second-moment method to bound such divergence. We first have the following bound that connects the KL divergence and the MMSE, that is proved in the \ifdefined\LONG appendix \else long version of the paper \fi using the I-MMSE relation.
\begin{lemma}
\label{le:i_mmse_rel}
If $\frac{1}{\lambda_n} D(\Q_{\lambda_n, n} || \Q_{0, n}) \in o(\tau_n)$ then for any $\beta > 1$ constant, $$\MMSE_n(\beta \lambda_n) \in o(\tau_n).$$
\end{lemma}
We can now bound the KL divergence $\frac{1}{\lambda_n} D(\Q_{\lambda_n, n} || \Q_{0, n})$ conditioning on a high probability event defined as follows.
\begin{definition}
\label{def:high_prob_event}
Let $\Q_{xy} = \P_n \otimes \Q_{y|x}$ the joint probability distribution of the vectors $(\bX, \bY)$ of problem defined in \Cref{eq:model} with snr $\lambda_n$. A series of events $\Omega_n \subset supp(\P_n) \otimes supp(\Q_{\lambda_n,n})$ occurs with high probability $1-o(\tau_n)$ uniformly over $\bX$ if
\begin{equation}
\Q_{xy}[\Omega_n | \bX] = 1 - o(\tau_n)
\end{equation}
for any $\bX \in supp(\P_n)$.
\end{definition}
Let us define as $\tilde{\Q}_{\lambda,n}$ the probability distribution of $\bY$ condition on a high probability event $\Omega_n$. Then the following bound holds.
\begin{lemma}
\label{le:kl_bound_conditioned}
If $\Omega_n$ is an event that occurs with uniform high probability $1 - o(\tau_n)$ then 
\begin{equation*}
\frac{1}{\lambda_n} D(\Q_{\beta \lambda_n, n} || \Q_{0, n}) \le \frac{1}{\lambda_n} D(\tilde{\Q}_{\beta \lambda_n, n} || \Q_{0, n}) + o(\tau_n)
\end{equation*}
\end{lemma}
To complete the proof, we need a claim that relates the KL given by the conditional distribution $\tilde{\Q}_{\lambda, n}$ of $\bY$ given a high probability event, to the overlap rate function of the problem. We first introduce the high probability events as follows.

\begin{lemma}
\label{le:omega_n}
Given a sequence $C_n$ with $\frac{1}{C_n} e^{-C_n^2/2} \in o(\tau_n)$, the event $$\Omega_n = \{ (\bX,\bY) \colon |\langle \bX, \bY \rangle - \sqrt{\lambda_n}| \le C_n\}$$ satisfies \Cref{def:high_prob_event}.
\end{lemma}
We can hence prove the following.
\begin{lemma}
\label{le:bound_conditioned_kl_rate}
For any $\epsilon > 0$ constant, if $\lambda_n^{-1/2+\epsilon} \in o(\tau_n)$, conditioning on the event $\Omega_n = \{ (\bX, \bY) \colon |\langle \bX, \bY \rangle - \sqrt{\lambda_n}| \le \sqrt{\log \lambda_n} \}$, the following bound holds:
\begin{equation*}
\frac{1}{\lambda_n} D(\tilde{\Q}_{\beta \lambda_n, n} || \Q_{0, n}) \le \sup_{t \in [0, 1]} \left[ \frac{t}{t+1} - \frac{r_n(t)}{2} \right] + o(\tau_n).
\end{equation*}
\end{lemma}
We can then finally bound the KL divergence conditioning on the events $\Omega_n$ defined in \Cref{le:bound_conditioned_kl_rate} as 
\begin{align*}
\frac{1}{\lambda_n} D(\Q_{\lambda_n, n} || \Q_{0, n}) &\le \frac{1}{\lambda_n} D(\tilde{\Q}_{\lambda_n, n} || \Q_{0, n}) + o(\tau_n) \\
&\le  \sup_{t \in [0, 1]} \left[ \frac{t}{t+1} - \frac{r_n(t)}{2} \right] + o(\tau_n) \\
&\in o(\tau_n)
\end{align*}
where the first inequality comes from \Cref{le:kl_bound_conditioned}, \Cref{le:omega_n} and the assumption that $\lambda_n^{-1/2 + \epsilon} \in o(\tau_n)$, the second inequality comes from \Cref{le:bound_conditioned_kl_rate} and the final inclusion is due to the assumption of the theorem on the rate function $r_n(t)$. The missing proofs of the lemmas are \ifdefined\LONG postponed to the appendix. \else given in the long version of the paper.\fi
\end{proof}

%% file: sections/application_tensor_PCA.tex
We here study the MLE, showing that a weaker AoN phenomenon extends to the behaviour of this estimator in the case of the sparse tensor-PCA model. 

\begin{definition}
\label{def:MLE}
The MLE for the generic model in \Cref{eq:model} is the estimator that maximizes the likelihood as
\begin{equation*}
\bX_{\MLE} = \argmax\limits_{\hat{\bX}(\bY) \in \bar{\mathcal{S}}_{n-1}} \Q_{y|x} (\bY | \bX(\bY))
\end{equation*}
\end{definition}
The following characterization of the MLE follows easily from the definition.
\begin{lemma}
\label{le:mle}
The MLE minimizes the probability of error $$\EP_n(\hat{\bX}) \coloneqq \P_{n} \Q_{y|x}[\hat{\bX}(\bY) \neq \bX] = \E\left[\mathbbm{1}_{\{\| \hat{\bX}(\bY) - \bX\|^2 >0\}} \right]$$
\end{lemma}
According to the latter lemma and \Cref{def:estimator}, we hence characterize the MLE as the optimal estimator according to the 0-1 loss, hence we can denote
$\bX_{\zerone} = \bX_{\MLE}$ and by
\begin{align*}
\MEP_n(\lambda) &= \P_{n} \Q_{y|x}[\bX_{\MLE}(\bY) \neq \bX] \\
&= \E\left[\mathbbm{1}_{\{\| \bX_{MLE} - \bX\|^2 >0\}} \right]    
\end{align*}
the minimum error probability obtained by such estimator. Based on the same definition, the constrained version is further defined. 

\subsection{Application to the sparse tensor-PCA problem}
In the following we assume for $d \ge 2$ the following sparse tensor-PCA model with observations 
\begin{equation}
\label{eq:model_PCA}
\bY = \sqrt{\lambda} \bx^{\otimes d} + \bZ,
\end{equation}
that corresponds to the additive Gaussian model defined in \Cref{eq:model} using $\bX = \bx^{\otimes d}$, with $\bx \in \reals^p$ and $n = p^d$, and considering the Frobenius norm for tensors in $\reals^n$ \footnote{Note that this problem can also encompass the planted problem in hypergraph, with observations in the \emph{upper-triangular part} of the tensor as 
\begin{equation*}
\bY = (\sqrt{\lambda} \bx^{\otimes d} + \bZ) \mathbbm{1}_{\{i_1 < \dots < i_d\}}
\end{equation*}
and with $n = \binom{p}{d}$. Results easily extend to the hypergraph variation seamlessly.}. A discrete uniform prior $\tilde{\P}_p$ over $\mathcal{S}_{p-1}$ induces a discrete uniform prior $\P_n$ over $\mathcal{S}_{n-1}$, hence the assumption of the model defined in \Cref{eq:model} are satisfied. Here and in the following we assume $\tilde{\P}_p$ to be a Bernoulli prior over the subset of the unit sphere with $k$ binary entries, hence 
$$\bx \in \left\{0,\frac{1}{\sqrt{k}}\right\}^p \cap \mathcal{S}_{p-1}=\mathcal{C}_{p,k} = supp(\tilde{\P}_p)$$
The cardinality of the hypothesis space is hence $M_p = \binom{p}{k}$, and $\lambda_n = \log M_p = k \log\left(\frac{p}{k}\right)(1 + o(1))$. The prior $\tilde{\P}_p$ maps to the uniform prior $\P_n$ over the space $supp(\P_n) \subsetneq \mathcal{C}_{n, s} =\{0, s^{-1/2}\}^n \cap \mathcal{S}_{n-1}$, where $s=k^d$. Note here the difference between the $supp(\P_n)$, that is the set of tensors formed as $\bx^{\otimes d}$ with $\bx \in \mathcal{C}_{p,k}$, and the set $\mathcal{C}_{n,s}$, that is the set of tensors with \emph{any} $s$ entries equal to $s^{-1/2}$. We here study the constrained estimators $\CMMSE$ and $\CMEP$ on the set $\mathcal{C}_{n,s}$, as it allows an easy characterization in terms of the unconstrained one. The main theorem of this section gives a sufficient condition for an AoN phenomenon to hold for the MEP.
\begin{theorem}
\label{th:aon_MEP}
For the sparse Bernoulli tensor-PCA model defined in \Cref{eq:model_PCA}, with $k \in o\left(\log^{\frac{1}{4d -1 }} p \right)$, the $\MEP$ satisfies the weak AoN transition as:
\begin{align*}
&\liminf\limits_{n\to \infty} \MEP_n(\beta \lambda_n)  \ge \frac{1}{4} \ \qquad \beta < 1 \\
&\lim_{n\to \infty} \MEP_n(\beta \lambda_n)
\ \ =  0 \ \ \qquad \beta > 1
\end{align*} 
The same transition holds for the $\CMEP_n(\beta \lambda_n)$.
\end{theorem}
We conjecture that the MLE undergoes a strict AoN transition, but further work is necessary to establish the full characterization of the MLE in the impossibility regime.
\begin{proof}
The main idea of the proof is to relate the MMSE to the MEP studying the constrained counterpart of both. For these latter two quantities, a simple first-moment method can be applied, as there exists a minimum non-vanishing distance between any two points in the constrained set $\mathcal{C}_{n,s}$. Hence, for any estimator $\hat{\bX}(\bY)$, using the Markov inequality we can derive the following bound on its error probability:
\begin{align*}
\EP_n(\hat{\bX}) = \E\left[\mathbbm{1}_{\{\|\hat{\bX} - \bX \|^2 > 0 \}}\right] &= \P\left[\| \hat{\bX} - \bX \|^2 \ge \frac{2}{s}\right] \\
&\le \frac{s}{2} \E\left[\| \hat{\bX} - \bX \|^2\right]
\end{align*}
Note that this is only possible if the estimator is constrained in $\mathcal{C}_{n,s}$. The following bounds can hence be derived with the full proof given in the appendix.
\begin{lemma}
\label{le:mmse_mep_relations}
Given the problem in \Cref{eq:model_PCA}, the following bounds hold:
\begin{equation}
\label{eq:ub_p_error}
\frac{s}{2} \CMMSE_n(\lambda) \ge 
\CMEP_n(\lambda)
\end{equation}
\begin{align}
\label{eq:lb_p_error}
\CMEP_n(\lambda) &\ge \frac{1}{4} \CMMSE_n(\lambda)^2 \\
\MEP_n(\lambda) &\ge \frac{1}{4} \MMSE_n(\lambda)^2
\end{align}
\end{lemma}
The inequality given in \Cref{eq:ub_p_error} relates now the $\CMMSE$ to $\CMEP$, such that if the $\CMMSE$ is small enough, then the $\MEP$ is small too. We can further have a bound that relates the $\MMSE$ to the $\CMMSE$ in the same direction, so as to derive a chain of inequalities between the $\MMSE$ and the $\MEP$. This is given by the following lemma.
\begin{lemma}
\label{le:mmse_cmmse}
For any $\epsilon > 0$, $\MMSE_n(\lambda) < \epsilon$ if and only if $ \CMMSE_n(\lambda) < 4 \epsilon s$.
\end{lemma}
\begin{proof}
Let us define the (constrained-) MSE distance as, respectively, $$\CMSE_n(\bX, \bY) = \| \bX_{\Cl_2}(\bY) - \bX \|^2$$ and 
$$\MSE_p(\bX, \bY) = \| \bX_{l_2}(\bY) - \bX \|^2.$$
We can write the expectation conditioning on the event $A$ that the first distance is smaller then a given $\delta < \frac{1}{2 s}$,
$$A = \{(\bX, \bY) \colon \CMSE_n(\bX, \bY) \le \delta \}$$ 
as:
\begin{align}
\label{eq:cmmse_decomposition}
\CMMSE_n(\lambda) &= \E\left[\CMSE_n(\bX, \bY)\right]\nonumber \\
&=\E\left[\CMSE_n(\bX, \bY) | A \right] \P[A] + \nonumber\\
& \hspace{1cm} + \E\left[\CMSE_n(\bX, \bY) | A^c \right] \P[A^c]. 
\end{align}
We now characterize the optimal constrained $l_2$ estimator as the simple rounding of the top entries of the standard posterior average estimator.
\begin{lemma}
\label{le:top_k}
The optimal estimator constrained in the hypothesis space $\mathcal{C}_{n,s}$ for the problem in \Cref{eq:model_PCA} for the MSE reads
$$\bX_{\Cl_2} = \argmin\limits_{\hat{\bX}(\bY) \in \mathcal{C}_{n, s}} \E[\| \hat{\bX}(\bY) - \bX\|^2] = \Top_s\left( \E[\bX | \bY] \right),$$
where the $\Top_s(\cdot)$ operator rounds the top $s$ entries of $\bX$ to $s^{-1/2}$, and zeros out all other entries.
\end{lemma}
Now we can easily note that the following geometrical lemma:
\begin{lemma}
\label{le:top_k_geometric}
For any integer $n$ and $s$ and $\bU \in \mathcal{C}_{n,s}$, $\bV \in \left[0, s^{-1/2}\right]^n$ and $\delta < \frac{1}{2 s}$, such that $\|\bU - \bV\|^2 \le \delta$, $$\Top_s(\bV) = \bU.$$
\end{lemma}
Combining \Cref{le:top_k} and \Cref{le:top_k_geometric} it follows that 
\begin{equation}
\label{eq:cmse_zero}
\E\left[ \CMSE_n(\bX, \bY) | A \right] = 0    
\end{equation} 
as $\Top_s(\bX_{\Cl_2}) = \bX_{\Cl_2}$. 
Using the same decomposition for the MMSE with respect to the event 
$$B = \{(\bX, \bY) \colon \MSE_n(\bX, \bY) \le \delta \}$$ 
we get
\begin{align}
\label{eq:mmse_decomposition}
\MMSE_n(\beta) &= \E[\MSE_n(\bX,\bY) | B ] \P[B] + \nonumber \\
&\hspace{1cm} +\E[\MSE_n(\bX,\bY) | B^c] \P[B^c] < \epsilon
\end{align}
Bounding as $\E[\MSE_n(\bX, \bY)| B] \ge 0$, $\P[B] \ge 0$ and $\E[\MSE_n(\bX,\bY) | B^c ] \ge \delta$ we get from \Cref{eq:mmse_decomposition}
\begin{equation}
\label{eq:b_c}
\P[B^c] < \frac{\epsilon}{\delta}.
\end{equation}
From \Cref{le:top_k}, we have further that for $\delta < \frac{1}{2 s}$,
\begin{align*}
B &\subset \{\CMSE_n(\bX,\bY) = 0 \} \subset A
\end{align*}
hence that 
\begin{equation}
\label{eq:a_c}
\P[A^c] \le \P[B^c].
\end{equation}
Plugging Equations \ref{eq:cmse_zero}, \ref{eq:b_c} and \ref{eq:a_c} into the decomposition in \Cref{eq:cmmse_decomposition} and using the fact that $\CMSE_n(\bX, \bY) \le 2$ we finally get
$$\CMMSE_n(\lambda) \le \frac{2 \epsilon}{\delta}.$$ The theorem follows from the arbitrariness of $\delta < \frac{1}{2 s}$.
\end{proof}
Plugging in the result of \Cref{le:mmse_cmmse} and \Cref{eq:ub_p_error} we get the further lemma that relates the MMSE in the recovery regime to the C-MMSE and the MEP. 
\begin{lemma}
$\MMSE_n(\lambda_n) \in o(1/s)$ if and only if $\CMMSE_n(\lambda_n) \in o(1)$. If $\MMSE_n(\lambda_n) \in o(1/s^2)$ then $\MEP_n(\lambda_n) \in o(1)$
\end{lemma}
In the same regime $\beta > 1$, we can hence use the results on the generalized AoN, \Cref{th:rate}, to have the MMSE to be $o(1/s^2)$. For such theorem to hold, we use the following lemma on the overlap rate function of the sparse tensor-PCA problem.
\begin{lemma}[Proof given in Proposition 3, \cite{niles2020all}]
\label{le:rate_bernoulli_pca}
For the Bernoulli sparse tensor-PCA problem with signal $\bX = \bx^{\otimes d}$, $d \ge 2$, and $\bx \in \{0, 1/\sqrt{k}\}^p \cap \mathcal{S}_{p-1}$ the following bound on the overlap rate function of the tensors $\bX, \bX'$ holds for any $t \in [0, 1]$:
\begin{equation*}
r_n(t) \ge \sqrt{t} - \frac{\mathcal{O}(1)}{\lambda_n}
\end{equation*}
\end{lemma}
Combining \Cref{le:rate_bernoulli_pca}, \Cref{th:rate}, \Cref{le:mmse_cmmse} and \Cref{le:mmse_mep_relations} we finally get the claim of the main theorem as $$\sqrt{t} \ge \frac{2t}{1+t}$$ for $t \in [0, 1]$ and as $k \in o\left(\log^{\frac{1}{4d -1 }} p \right)$ implies $$\lambda_n^{-1/2 + \epsilon} = (k \log(p/k) + o(1))^{-1/2 + \epsilon} \in o(1/s^2) = o(1/k^{2d})$$
\end{proof}
Analogously, we can prove the following transition for the constrained MMSE.
\begin{theorem}
\label{th:aon_CMMSE}
For the sparse Bernoulli tensor PCA model with $k \in o\left(\log^{\frac{1}{2d -1 }} p \right)$, the $\CMMSE$ satisfies the AoN transition 
\begin{align*}
&\liminf\limits_{n\to \infty} \CMMSE_n(\beta \lambda_n)  \ge 1 \ \qquad \beta < 1 \\
&\lim_{n\to \infty} \CMMSE_n(\beta \lambda_n)
\ \ =  0 \ \ \qquad \beta > 1
\end{align*} 
\end{theorem}
\begin{proof}
The proof follows the same steps of the proof of \Cref{th:aon_MEP}, with only two main differences. First, in the impossibility regime, we can use the stronger bound $\CMMSE_n(\lambda) \ge \MMSE_n(\lambda)$ in place of \Cref{eq:lb_p_error} to get the first part of the theorem. In the recovery regime, we can use the weaker requirement $\MMSE_n(\beta \lambda_n) \in o(1/s)$ that is satisfied by the assumption of the theorem $k \in o\left(\log^{\frac{1}{2d -1 }} p \right)$.
\end{proof}

%% file: sections/conclusion.tex
In this paper, we analysed the maximum-likelihood estimator for the sparse tensor-PCA problem with Bernoulli prior. We established that this estimator undergoes a weak AoN transition and conjectured that this transition is equivalent to the MMSE transition.
The proof follows from the connection of the MLE to the MMSE using the first and second-moment method in the constrained signal space, and hence it is of independent interest as it can lead to further results from the community. 

While this paper sets a first step in understanding a wider range of optimal estimators in sparse high-dimensional inference problems, a general theory of the all-or-nothing statistical transition is still lacking. This theory could provide a wider understating of the phenomenon, including the analysis of vectorial estimators for planted matrix and tensor PCA problems, that is not here considered and is carried out in the dense setting using rigorous tools of statistical physics and replica methods (see \cite{lesieur2017statistical}).

The same methods established recently that tractable estimators, like the approximate-message-passing algorithms, undergo the same all-or-nothing transition in the sparse matrix PCA problem \cite{barbier2020all}. The extension of these results to the sparse tensor-PCA problem and currently optimal algorithms for this problem, like averaged gradient descent \cite{biroli2020iron} and sum-of-squares algorithms \cite{hopkins2015tensor}, is of crucial importance.

%% file: appendix/postponed_proofs_aon.tex
\begin{proof}[Proof of \Cref{le:i_mmse_rel}]
\label{proof:i_mmse_rel}
Using, respectively, the maximum entropy bound and 1/2 Lipschitz-continuity of the function (see \Cref{le:properties_kl} and \cite{niles2020all}) we can see first that
\begin{align*}
0 \le \frac{1}{\lambda_n} D(\Q_{\beta \lambda_n, n} || \Q_{0, n}) - \frac{1}{2}(\beta - 1) \le \\
& \hspace{-2cm} \le \frac{1}{\lambda_n} D(\Q_{\lambda_n, n} || \Q_{0, n}) \in o(\tau_n)
\end{align*}
where the first inequality comes from the maximum entropy bound, the second from Lipschitz-continuity, and the inclusion follows from the assumption of the theorem. For convenience let us denote $$\frac{1}{\lambda_n} D(\Q_{\beta \lambda_n, n} || \Q_{0, n}) = \frac{1}{2} (\beta -1) + f_n(\beta)$$ where $\lim_{n \to \infty} \frac{f_n(\beta)}{\tau_n} = 0$ for any $\beta > 1$. 
We can now use the I-MMSE relation (see \cite{verdu1994generalizing, niles2020all}), such that $\MMSE_n(\beta \lambda_n) = 1 - 2 \frac{d}{d \beta} \frac{1}{\lambda_n} D(\Q_{\beta \lambda_n, n} || \Q_{0, n})$.
\begin{align*}
\MMSE(\beta \lambda_n) &= 1 - 2 \frac{d}{d \beta} \left(\frac{1}{2} (\beta -1) + f_n(\beta)\right) \\
&= - 2 \frac{d}{d \beta} f_n(\beta)
\end{align*}
hence 
\begin{align*}
\lim_{n \to \infty} \frac{\MMSE(\beta \lambda_n)}{\tau_n} &= \lim_{n \to \infty} - 2 \frac{1}{\tau_n} \frac{d}{d \beta} f_n(\beta) \\
&= \lim_{n \to \infty} - 2 \frac{d}{d \beta} \frac{1}{\tau_n} f_n(\beta) = 0
\end{align*}
where the second equality follows from the linearity of differentiation and the third from the interchanging of limit and differentiation under uniform convergence.
\end{proof}

\begin{proof}[Proof of \Cref{le:kl_bound_conditioned}]
Following the proof of Theorem 5 in \cite{banks2018information}, and defining the function $Z(\bY) = \frac{\Q_{\beta \lambda_n, n}(\bY)}{\Q_{0 , n}(\bY)}$ we have that 
\begin{align}
\label{eq:bound_Z}
D(\tilde{\Q}_{\beta \lambda_n, n} || \Q_{0, n}) - D(\Q_{\beta \lambda_n, n} || \Q_{0, n}) \ge \nonumber\\
&\hspace{-3.5cm} \ge \E_{\tilde{\Q}_{\beta \lambda_n, n}} \log Z(\bY) - \E_{\Q_{\beta \lambda_n, n}} \log Z(\bY)
\end{align}
Using the definition of the conditional pdf, we have
\begin{align}
\label{eq:conditional_q}
\tilde{\Q}_{\beta \lambda_n, n}(\bY) &= \Q_{\beta \lambda_n, n}(\bY | \Omega_n) = \E_{\P_n} \Q_{y|x}(\bY | \Omega_n, \bX) \nonumber \\
&= \E_{\P_n} \frac{\Q_{y|x}(\bY | \bX) \Q_{xy}[\Omega_n | \bX, \bY]}{\Q_{xy}[\Omega_n | \bX]} \nonumber\\
&= \E_{\P_n} \frac{\Q_{y|x}(\bY | \bX) \mathbbm{1}_{\Omega_n}(\bX,\bY)}{\Q_{xy}[\Omega_n | \bX]},
\end{align}
where $\mathbbm{1}_A(\cdot)$ is the indicator function of set $A$. Plugging \Cref{eq:conditional_q} into \Cref{eq:bound_Z} we get
\begin{align}
\label{eq:lb_kl_diff}
&D(\tilde{\Q}_{\beta \lambda_n, n} || \Q_{0, n}) - D(\Q_{\beta \lambda_n, n} || \Q_{0, n}) \ge \nonumber\\
& \ge \E_{\P_n} \frac{\E_{\Q_{y|x}} \left[ (\mathbbm{1}_{\Omega_n}(\bX,\bY) -  \Q_{xy}[\Omega_n | \bX] ) \log Z(\bY) \right]}{\Q_{xy}[\Omega_n | \bX]}
\end{align}
Using the Cauchy-Schwartz inequality we can bound the expectation over $\bY$ as 
\begin{align}
\label{eq:lb_cs}
&\left|\E_{\Q_{y|x}} \left[ (\mathbbm{1}_{\Omega_n}(\bX,\bY) -  \Q_{xy}[\Omega_n | \bX] ) \log Z(\bY) \right]\right| \le \nonumber\\
&\sqrt{\E_{\Q_{y|x}} (\mathbbm{1}_{\Omega_n}(\bX,\bY) -  \Q_{xy}[\Omega_n | \bX] )^2 \cdot  \E_{\Q_{y|x}} \left[ \log^2 Z(\bY) \right] }
\end{align}
It is easy to see that
\begin{align*}
\E_{\Q_{y|x}} &\left[ (\mathbbm{1}_{\Omega_n}(\bX,\bY) -  \Q_{xy}[\Omega_n | \bX] )^2 \right] = \\  &= \E_{\Q_{y|x}} \big[ (\mathbbm{1}_{\Omega_n}(\bX,\bY) +  \Q_{xy}[\Omega_n | \bX]^2 + \\
&\hspace{1cm} - 2 \cdot \mathbbm{1}_{\Omega_n}(\bX,\bY) \Q_{xy}[\Omega_n | \bX]^2 \big] \\
&=\E_{\Q_{y|x}} \left[ \mathbbm{1}_{\Omega_n}(\bX,\bY) \right] + \Q_{xy}[\Omega_n | \bX]^2 + 
\\& \hspace{1cm} - 2 \cdot \E_{\Q_{y|x}} \left[ \mathbbm{1}_{\Omega_n}(\bX,\bY) \right] \Q_{xy}[\Omega_n | \bX] = 
\\&= Q_{xy}[\Omega_n | \bX] - \Q_{xy}[\Omega_n | \bX]^2
\end{align*}
hence recombining the latter and \Cref{eq:lb_kl_diff} and \Cref{eq:lb_cs} we get
\begin{align*}
D(&\tilde{\Q}_{\beta \lambda_n, n} || \Q_{0, n}) - D(\Q_{\beta \lambda_n, n} || \Q_{0, n}) \ge \\
&\ge - \E_{\P_n} \sqrt{\frac{1- \Q_{xy}[\Omega_n | \bX]}{\Q_{xy}[\Omega_n | \bX]} \E_{\Q_{y|x}} \left[\log^2 Z(\bY) \right] }.
\end{align*}
Using again the Cauchy-Schwartz inequality over the expectations on $\P_n$ we finally get 
\begin{align*}
D(&\tilde{\Q}_{\beta \lambda_n, n} || \Q_{0, n}) - D(\Q_{\beta \lambda_n, n} || \Q_{0, n}) \ge \\
&- \sqrt{\E_{\P_n} \left(\frac{1- \Q_{xy}[\Omega_n | \bX]}{\Q_{xy}[\Omega_n | \bX] }\right)^2 } \cdot \sqrt{ \E_{\P_n} \E_{\Q_{y|x}} \left[\log^2 Z(\bY) \right]} \\
&= - o(\tau_n) \cdot \sqrt{\E_{\Q_{\beta \lambda_n, n}} \left[\log^2 Z(\bY) \right]}
\end{align*}
where the equality comes from the assumption on the event $\Omega_n$. Using now the result from Proposition 3 in \cite{niles2020all}, 
\begin{equation*}
\sqrt{\E_{\Q_{\beta \lambda_n, n}} \left[\log^2 Z(\bY) \right]} = \mathcal{O}(\log M_n)
\end{equation*}
we get the claim.
\end{proof}

\begin{proof}[Proof of \Cref{le:omega_n}]
We can see easily that
\begin{align*}
\Q_{xy}[\Omega_n | \bX] &= \Q_{y|z}[\Omega_n | \bX] \\
&= \P_z [|\langle \bX, \bZ \rangle| \le C_n] \\
&= 2 \phi(C_n) - 1 \\
&= 1 - \sqrt{\frac{2}{\pi}} \frac{1}{C_n} e^{-C_n^2/2}\left(1 + \mathcal{O}\left(\frac{1}{C_n}\right)\right)
\end{align*}
where in the second equality we used the fact that $\bY = \sqrt{\lambda_n} \bX + \bZ$ and that  $\| \bX \|^2 = 1$, the third follows from the fact that $\langle \bX, \bZ \rangle$ is a univariate Gaussian random variable distributed as $\mathcal{N}(0, 1)$, and $\phi(\cdot)$ is the cdf of the standard Gaussian, with asymptotics given, for large $x$, as $\phi(x) = 1 - \frac{1}{\sqrt{2 \pi} x} e^{-x^2/2}\left(1+\mathcal{O}\left(\frac{1}{x^2}\right)\right)$. We hence have the claim as $$1 - \Q_{xy}[\Omega_n | \bX] = \mathcal{O}\left(\frac{1}{C_n} e^{-C_n^2/2}\right) \in o(\tau_n).$$
\end{proof}

\begin{proof}[Proof of \Cref{le:bound_conditioned_kl_rate}]
Using the Jensen inequality, it can be easily seen that for any two distributions 
\begin{align*}
D(p || q) = \E_p \log \frac{p(x)}{q(x)} \le \log\left( \E_p \frac{p(x)}{q(x)}\right) &= \log\left( \E_q \left(\frac{p(x)}{q(x)} \right)^2 \right)
\end{align*}
To bound the KL, we can hence study the ratio $p/q$. Using \cref{eq:conditional_q} we can hence write
\begin{equation*}
\frac{\tilde{\Q}_{\beta \lambda_n, n}(\bY)}{\Q_{0, n}(\bY)} = \E_{\P_n} \frac{1}{\Q_{xy}[\Omega_n | \bX]} \frac{\Q_{\beta \lambda_n, n}(\bY | \bX)}{\Q_{\beta 0, n}(\bY)} \mathbbm{1}_{\Omega_n}(\bX,\bY).
\end{equation*}
Plugging in the definition of the model for $\bY$ we can easily see that 
\begin{align*}
\frac{\Q_{\beta \lambda_n, n}(\bY | \bX)}{\Q_{ 0, n}(\bY)} &= \frac{\exp(-\frac{1}{2} \| \bY - \sqrt{\lambda_n}\bX \|^2)}{\exp(-\frac{1}{2} \| \bY \|^2)} \\
&=\exp\left(\sqrt{\lambda_n} \langle \bX, \bY \rangle - \frac{\lambda_n}{2} \right)
\end{align*}
where in the second inequality we used the fact that $\|\bX\|^2 = 1$. Using the latter, we can obtain
\begin{align*}
\left(\frac{\tilde{\Q}_{\beta \lambda_n, n}(\bY | \bX)}{\Q_{ 0, n}(\bY)}\right)^2 = \\
& \hspace{-3.3cm} = \E_{\P_n^{\otimes 2}}  \frac{\exp\left(\sqrt{\lambda_n} \langle \bX + \bX', \bY \rangle - \lambda_n \right)}{\Q_{xy}[\Omega_n | \bX] Q_{xy}[\Omega_n | \bX'] } \mathbbm{1}_{\Omega_n}(\bX,\bY) \mathbbm{1}_{\Omega_n}(\bX',\bY) 
\end{align*}
Using the fact that for $C_n = \sqrt{\log \lambda_n}$ and $\lambda_n^{-1/2 + \epsilon} \in o(\tau_n)$ we can satisfy the assumption of \Cref{le:omega_n}, we can exchange the small-o notation and the integrals as
\begin{align*}
\left(\frac{\tilde{\Q}_{\beta \lambda_n, n}(\bY | \bX)}{\Q_{ 0, n}(\bY)}\right)^2 = \\
& \hspace{-2cm} = (1 + o(\tau_n))\E_{\P_n^{\otimes 2}} \big[ \mathbbm{1}_{\Omega_n}(\bX,\bY) \mathbbm{1}_{\Omega_n}(\bX',\bY) \\
&\times \exp\left(\sqrt{\lambda_n} \langle \bX + \bX', \bY \rangle - \lambda_n \right)\big]
\end{align*}
We hence have a bound for the KL that reads
\begin{equation*}
\frac{1}{\lambda_n} D(\tilde{\Q}_{\beta \lambda_n, n} || \Q_{0, n}) \le \frac{1}{\lambda_n} \log \left[ \E_{\P_n^{\otimes 2}} m_n(\bX, \bX') \right] + o(\tau_n)
\end{equation*}
where $m_n$ is defined as 
\begin{align*}
m_n(\bX, \bX') \coloneqq \E_{\Q_{0, n}} \big[ \mathbbm{1}_{\Omega_n}(\bX,\bY) \mathbbm{1}_{\Omega_n}(\bX',\bY) \\
& \hspace{-3cm} \times \exp\left(\sqrt{\lambda_n} \langle \bX + \bX', \bY \rangle - \lambda_n \right) \big]
\end{align*}
and we used the fact that $\frac{o(\tau_n)}{\lambda_n} \in o(\tau_n)$ and where we used Fubini's theorem to exchange the order of the integrals. 
Note that $\Q_{0,n} = \P_z = \normal(0, 1)$. Now it is sufficient to prove that $$\frac{1}{\lambda_n} \log \left[ \E_{\P_n^{\otimes 2}} m_n(\bX, \bX') \right] \le \sup_{t \in [0, 1]} \left( \frac{t}{t+1} - \frac{r_n(t)}{2} \right) + o(\tau_n)$$ to get the claim. We can readily see that the function $m_n$ depends only on the overlap $\rho = \langle \bX, \bX' \rangle$ due to the rotational invariance of the Gaussian pdf. Using \Cref{le:ub_m_n}, the definition of $r_n$, the monotonicity of the exponential function and the simple inequality $$sup(f + g) \le \sup f + \sup g$$ we get
\begin{align*}
\frac{1}{\lambda_n} \log \E_{\P_n^{\otimes 2}} m_n(\rho) &\le \frac{\log(2 L_n)}{\lambda_n} + \\
&\hspace{-2.2cm} + \sup_{t \in [-1, 1]}\left( \left(  \frac{t}{1+t}\right)_+ - \frac{r_n(t)}{2} \right) + \frac{C_n}{\lambda_n^{1/2}} + \mathcal{O}\left(\frac{1}{L_n}\right)
\end{align*}
We can easily observe that the supremum can be limited to the interval $t \in [0, 1]$ noting that $r_n(-1) = 0$ and $r_n(t)$ is a non-negative function. 
The claim then follows easily from the assumption of $\lambda_n^{-1/2 + \epsilon} \in o(\tau_n)$ and choosing $L_n = \floor{\lambda_n^{1/2}}$ and $C_n = \sqrt{\log \lambda_n}$.
\end{proof}

%% file: appendix/postponed_proofs_mle.tex
\begin{proof}[Proof of \Cref{le:mle}]
\begin{align*}
\P_{n} \Q_{y|x}[\hat{\bX}(\bY) \neq \bX] &= \E_{P_n} \E_{\Q_{y|x}} \mathbbm{1}_{\{\hat{\bX}(\bY) \neq \bX\}} \nonumber \\
& = 1 - \E_{P_n} \E_{\Q_{y|x}} \mathbbm{1}_{\{\hat{\bX}(\bY) = \bX\}} \nonumber \\
&= 1 - \int d \bY \Q_{y|x}(\bY | \hat{\bX}(\bY)) \P_n(\hat{\bX}(\bY))
\end{align*}
hence $\argmin\limits_{\hat{\bX}(\cdot)} \P_{n} \Q_{y|x}[\hat{\bX}(\bY) \neq \bX]$ satisfies for every $\bY$ 
\begin{align*}
&\left(\argmin_{\hat{\bX}(\cdot)} \P_{n} \Q_{y|x}[\hat{\bX}(\bY) \neq \bX]\right)(\bY) =\\
&\hspace{2cm} = \argmax_{\hat{\bX}(\bY)} \Q_{y|x}(\bY | \hat{\bX}(\bY)) \P_n(\hat{\bX}(\bY)),
\end{align*} 
that, for uniform prior, corresponds to the MLE estimator.
\end{proof}

\begin{proof}[Proof of \Cref{le:mmse_mep_relations}]
The two bounds are, respectively, given by the first and second-moment methods. \Cref{eq:ub_p_error} follows easily from from the Markov inequality as for any estimator $\hat{\bX}(\bY)$,
\begin{align*}
\EP_n(\hat{\bX}) = \E\left[\mathbbm{1}_{\{\|\hat{\bX} - \bX \|^2 > 0 \}}\right] &= \Pr\left[\| \hat{\bX} - \bX \|^2 \ge \frac{2}{s}\right] \\
&\le \frac{s}{2} \E\left[\| \hat{\bX} - \bX \|^2\right]
\end{align*}
hence
\begin{align*}
\CMEP_n(\lambda) &= \min_{\hat{\bX}(\bY)\in \mathcal{C}_{n,s}} \E\left[\| \hat{\bX} - \bX \|_0 \right]  \\
& \le \frac{s}{2} \min_{\hat{\bX}(\bY)\in \mathcal{C}_{n,s}}  \E[\| \hat{\bX} - \bX \|^2] = \frac{s}{2} \CMMSE_n(\lambda).
\end{align*}
To prove the second bound, we use the Paley–Zygmund inequality that reads for a general \emph{positive} random variable $Z$ and $0 \le \theta \le \E[Z]$
\begin{equation}
\label{eq:paley_zygmund}
\Pr[Z > \theta] \ge \frac{(\E[Z] - \theta)^2}{\E[Z^2]}.
\end{equation}
Using \Cref{eq:paley_zygmund} for the random variable $\|\hat{\bX} - \bX \|^2$ we obtain
\begin{equation*}
\EP_n(\hat{\bX}) = \Pr\left[\|\hat{\bX} - \bX \|^2 > 0\right] \ge \frac{\E\left[\|\hat{\bX} - \bX \|^2\right]^2}{\E\left[\|\hat{\bX} - \bX \|^4\right]}
\end{equation*}
from which we get 
\begin{align*}
\MEP_n(\lambda) &= \min_{\hat{\bX}(\bY)\in \bar{\mathcal{S}}_{n-1}} \Pr\left[\|\hat{\bX} - \bX \|^2 > 0\right] \\
& \ge \min_{\hat{\bX}(\bY)\in \bar{\mathcal{S}}_{n-1}} \frac{\E[\|\hat{\bX} - \bX \|^2]^2}{\E[\|\hat{\bX} - \bX \|^4]} \\
&\ge \frac{\left(\min_{\hat{\bX}(\bY)\in \bar{\mathcal{S}}_{n-1}} \E[\|\hat{\bX} - \bX \|^2]\right)^2}{\max_{\hat{\bX}(\bY)\in \bar{\mathcal{S}}_{n-1}} \E[\|\hat{\bX} - \bX \|^4]} \\
&\ge \frac{1}{4} \MMSE_p(\lambda)^2
\end{align*}
where in the last inequality we used the definition of the $\MMSE$ and the fact that for any two vectors $$\bm{a}, \bm{b} \in \bar{\mathcal{S}}_{n-1}, \quad \|\bm{a} - \bm{b} \|^2 \le 2.$$ The third inequality follows analogously.
\end{proof}

\begin{proof}[Proof of \Cref{le:top_k}]
\begin{align}
\E[\| \hat{\bX} - \bX\|^2] = \E\left[ \|\hat{\bX}\|^2\right] + \E\left[ \|\bX\|^2\right] - 2 \E\left[ \sum_i^n  \hat{X}_i X_i\right].
\end{align}
Given the constraint on the estimator and on the fact that $\|\bX\|^2 = \|\hat{\bX}\|^2 = 1$, the optimization problem becomes: 
$$\argmax_{\hat{\bX}(\bY) \in \mathcal{C}_{n,s}} \E_{\Q_{\lambda, n}} \sum_{i=1}^n \hat{\bX}_i (\bY) \E[X_i|\bY]$$ hence for every fixed $\bY$ the optimal estimator reads
$$\argmax_{\hat{\bX}(\bY) \in \mathcal{C}_{n,s}} \sum_i^n \hat{\bX}_i(\bY) \E[\bX|\bY]_i.$$ The theorem follows easily from linearity and from $\mathcal{C}_{n, s}$ being a binary set for which the greedy algorithm is optimal. 
\end{proof}

\begin{proof}[Proof of \Cref{le:top_k_geometric}]
\begin{align*}
\delta \ge \|\bU - \bV\|^2 &= \sum_{i \colon U_i  = s^{-1/2}} \left(s^{-1/2}-V_i\right)^2 + \sum_{i \colon U_i  = 0} V_i^2 \\
&\ge \left(\max_{i \colon U_i  = s^{-1/2}}\left(s^{-1/2}-V_i\right)\right)^2 + \left(\max_{i \colon U_i  = 0} V_i\right)^2 \\
&= \left(s^{-1/2}-\min_{i \colon U_i  = s^{-1/2}} V_i\right)^2 + \left(\max_{i \colon U_i  = 0} V_i\right)^2
\end{align*}
Multiplying both sides of the latter inequality by $s$ we get a inequality of the form $\left(1-a\right)^2 + b^2 \le s \delta < \frac{1}{2}$ for $a,b \in \left[0, 1\right]$. It is easy to observe using simple calculus that this implies $a > b$ and hence $$\min_{i \colon U_i  = s^{-1/2}} V_i > \max_{i \colon U_i  = 0} V_i.$$ Given this condition, and the definition of the $\Top_s(\cdot)$ operator, it follows that $\Top_s(\bV) = \bU$. Note that the strict inequality $\delta < \frac{1}{2 s}$ is essential to guarantee the strict inequality above and hence that there are no ties in the selection of the top $s$ entries.
\end{proof}

%% file: appendix/appendix_useful_lemmas.tex
\begin{lemma}
\label{le:mmse}
The posterior average $\E[\bX | \bY]$ is the optimal estimator (MMSE) for the $l_2$ loss $l_2(\hat{\bx}, \bx) = \| \hat{\bX} - \bX\|^2$ and the generic model in \cref{eq:model}, formally
$$\argmin_{\hat{\bX}(\bY) \in \bar{\mathcal{S}}_{n-1}} \E[\| \hat{\bX}(\bY) - \bX\|^2] = \E[\bX | \bY]$$
\end{lemma}
\begin{proof}
\begin{align}
\frac{\partial}{\partial \hat{\bX}} \E[\| \hat{\bX} - \bX\|^2] &= \E_{\Q_{\lambda, n}} \sum_{\bX} \  \P_n(\bX | \bY) \frac{\partial}{\partial \hat{\bX}} \| \hat{\bX} - \bX\|^2 =\\
&= \E_{\Q_{\lambda, n}} \sum_{\bX} \P_n(\bX|\bY) 2 (\hat{\bX} - \bX) = 0
\end{align}
from which it follows easily that for every $\bY$, the gradient is equal to zero if $\hat{\bX}(\bY) = \E[\bx | \bY]$. It can be easily shown that the Hessian of the loss is positive semidefinite and hence satisfies the property of having a global minimum.
\end{proof}

We here further characterize the MLE estimator as following:
\begin{lemma}[equivalent to Theorem 1 in \cite{corinzia2019exact}]
For the model defined in \Cref{eq:model}, the MLE estimator reads: 
\begin{equation}
\label{eq:mle_estimator}
\bX_{\zerone}(\bY) = \argmax\limits_{\bX \in \bar{\mathcal{S}}_n} \sum_{i_1, \dots ,i_d} Y_{i_1, \dots, i_d} X_{i_1, \dots, i_d} 
\end{equation}
\end{lemma}
\begin{proof}
\begin{align*}
\log \Q_{y|x}&(\bY|\bX) = \sum_{i_1 , \dots , i_d} \log \Q_{y|x}(Y_{i_1, \dots, i_d} | x_{i_1} \cdot \dots \cdot x_{i_d}) \nonumber \\
&=\sum_{i_1 , \dots , i_d} - \frac{1}{2} \log 2 \pi - \frac{1}{2} (Y_{i_1, \dots, i_d} - \beta x_{i_1} \cdot \dots \cdot x_{i_d})^2 \nonumber \\
&= \sum_{i_1 ,\dots , i_d} - \frac{1}{2} \log 2 \pi - \frac{1}{2} \beta^2 x_{i_1}^2 \cdot \dots \cdot x_{i_d}^2 +\nonumber \\
& \hspace{2cm} - \frac{1}{2} Y_{i_1, \dots, i_d}^2 + \beta Y_{i_1, \dots, i_d} x_{i_1} \cdot \dots \cdot x_{i_d} \nonumber \\
&= - \frac{n}{2} \log 2 \pi - \frac{1}{2} \lambda d! \binom{k}{d} - \frac{1}{2} \sum_{i_1 , \dots , i_d} Y_{i_1, \dots, i_d}^2 +\nonumber \\ 
& \hspace{2cm } + \sqrt{\lambda} \sum_{i_1 , \dots , i_d} Y_{i_1, \dots, i_d} x_{i_1} \cdot \dots \cdot x_{i_d}.
\end{align*}
The theorem follows easily noting that only the last term depends on $\bX$. 
\end{proof}

\begin{lemma}
\label{le:properties_kl}
Given the setting in \Cref{eq:model}, for all $n$ and $\lambda >0$, the function $\beta \to \frac{1}{\lambda} D(\Q_{\beta \lambda,n}|| \Q_{0,n})$ is nonnegative, nondecreasing, 1/2-Lipschitz and satisfies the bound
$$\frac{1}{\lambda} D(\Q_{\beta \lambda,n}|| \Q_{0,n}) \ge \frac{1}{2} - \frac{\log M_n}{\lambda}.$$
\end{lemma}
\begin{proof}
The proof is given in Lemma~2 and Lemma~3 in \cite{niles2020all}.
\end{proof}

\begin{lemma}
\label{le:ub_m_n}
Given the setting of the problem defined in \Cref{eq:model_PCA} and the function 
\begin{align*}
m_n(\bX, \bX') \coloneqq \E_{\Q_{0, n}} \big[ \mathbbm{1}_{\Omega_n}(\bX,\bY) \mathbbm{1}_{\Omega_n}(\bX',\bY) \\
& \hspace{-3cm} \times \exp\left(\sqrt{\lambda_n} \langle \bX + \bX', \bY \rangle - \lambda_n \right) \big],
\end{align*} there exist a constant $C > 0$ such that for any integer sequence $L_n$ the following bound holds:
\begin{align*}
\E_{\P_n^{\otimes 2}} m_n(\rho) &\le 2 L_n \sup_{t \in [-1, 1]} \exp\bigg( \lambda_n \left(  \frac{t}{1+t}\right)_+ + \\
& \hspace{1cm} + \log \P[\rho \ge t] + C_n \lambda_n^{1/2} + \mathcal{O}\left(\frac{\lambda_n}{L_n}\right)\bigg)
\end{align*}
\end{lemma}
\begin{proof}
The proof follows, \emph{mutatis mutandis}, the proof of Theorem 4 in \cite{niles2020all}, with the minor change of the definition of the events $\Omega_n$ that are here defined as $\{|\langle \bX, \bZ\rangle - \sqrt{\lambda_n}| \le C_n\}$.
\end{proof}